\documentclass[twocolumn,showpacs,preprintnumbers,amsmath,amssymb,pra]{revtex4}

\usepackage{graphicx}% Include figure files
\usepackage{epstopdf}
\usepackage{dcolumn}% Align table columns on decimal point
\usepackage{bm}% bold math
\usepackage{subfigure}
\usepackage{amsthm}
\usepackage{color}
\newcommand{\ket}[1]{\ensuremath{\left|\left. #1\right>\right.}}

\newcommand{\braket}[2]{\ensuremath{\left<\left.#1\left| #2\right.\right.\right>}}

\newcommand{\denmat}[1]{|#1\rangle\langle #1|}
\newtheorem{theorem}{Theorem}
\newtheorem{definition}{Definition}
\newtheorem{lemma}{Lemma}

\newtheorem{prop}{Proposition}

\begin{document}

%\preprint{}
\title{Mixture of entangled pure states with maximally mixed one-qudit reduced density matrices}
\author{ Marvin M. Flores}
\author{Eric A. Galapon}
\email{eagalapon@up.edu.ph}
\affiliation{Theoretical Physics Group, National Institute of Physics, University of the Philippines, Diliman Quezon City, 1101 Philippines}
\date{\today}
\begin{abstract}
We study the conditions when mixtures of entangled pure states with maximally mixed one-qudit reduced density matrices remain entangled. We found that the resulting mixed state remains entangled when the number of entangled pure states to be mixed is less than or equal to the dimension of the pure states. For the latter case of mixing a number of pure states equal to their dimension, we found that the mixed state is entangled provided that the entangled pure states to be mixed are not equally weighted. We also found that one can restrict the set of pure states that one can mix from in order to ensure that the resulting mixed state is genuinely entangled.
\end{abstract}
\pacs{03.67.Mn}
\maketitle
\section{Introduction}

One of the outstanding problems in the theory of quantum entanglement, which is of paramount importance when using entanglement as a resource for quantum information processing is its detection \cite{detection}. In other words, given an arbitrary quantum state, the problem of detection asks whether the state is entangled or not. The difficulty of the problem of entanglement detection is reflected in the abundance of separability criteria in literature \cite{detection, horodecki}, some involving extremizations and working only for some special cases. The problem is even more pronounced in the case of mixed states, where the typical recourse would be to construct convex roof extensions of existing measures for pure states, something which is hard to compute in general \cite{quantifying,entropy}. Hence, it would be very helpful if we can come up with classes of mixed states which are entangled, thus saving us from having to apply an entanglement detection scheme should a given state belong to these classes.

In this paper, we introduce a class of entangled mixed states by mixing distinct $N$-partite entangled pure states having equal dimensions $d$ and having maximally mixed one-qudit reduced density matrices. We are motivated by the fact that although a convex combination of separable states is again separable, the mixture of entangled pure states is not necessarily entangled as can be easily seen by mixing two distinct maximally entangled Bell states of equal weights \cite{popescu}. Our choice of mixing entangled pure states where all the one-qudit reduced density matrices are maximally mixed is justified since we utilized an entropic-based separability criterion to determine whether the resulting mixed states are entangled or not. Entropic-based separability criteria in general detect entanglement by virtue of how much information is present in the correlations between subsystems, rather than in the subsystems themselves.

We show that for the mixture to remain entangled, the dimension $d$ of these entangled pure states must exceed the number of states $M$ in the convex combination. Also, we show that when mixing $d$ number of these entangled pure states, the resulting mixed state is entangled provided that their weights are not all equal. We also show that we can construct mixed states which are genuinely entangled at the expense of making the set of pure states that we can mix from smaller.

The paper will begin by introducing the separability criterion based on purity in Section 2 which we will be utilizing in determining which states are entangled or not. In Section 3, we study the conditions wherein the mixture of entangled pure states with maximally mixed one-qudit reduced density matrices remain entangled. Then we restrict the set of these pure states that we can mix from in Section 4 with the advantage of having our resulting mixed state to have an entanglement that we know to be genuine. Section 5 discusses the consistency of our results with the measure of entangled mixed states for bipartite qubits known as concurrence. Finally, we draw our conclusions in Section 6.

\section{Purity-based separability criterion}

Let $\rho$ be a density matrix for a composite system. We say that $\rho$ is a \textbf{product state} if there exist states $\rho^A$ for Alice and $\rho^B$ for Bob such that
\begin{equation}
	\rho = \rho^A\otimes\rho^B.
\end{equation}
The state is called \textbf{separable}, if there are convex weights $p_i$ and product states $\rho_i^A\otimes\rho_i^B$ such that
\begin{equation}
	\label{eq:rhoseparable}
	\rho=\sum_i p_i \rho_i^A\otimes\rho_i^B
\end{equation}
holds. Otherwise, the state is called \textbf{entangled}. In general, for a mixed state of $N$ systems, separability entails that we could write the state as
\begin{equation}
	\label{eq:nrhoseparable}
	\rho=\sum_i p_i \rho_i^1\otimes\cdots\otimes\rho_i^N.
\end{equation}

In this paper, we also distinguish between $k$-separable states and genuinely multipartite entangled states. A state is $k$-separable if it factorizes into $k$ states $\rho_i$, each of which describes either one or several subsystems, that is, the $N$-partite state

\begin{equation}
\rho=\sum_i p_i \rho_i^1\otimes\cdots\otimes\rho_i^k
\end{equation}

\noindent is $k$-separable with $1 < k \leq N$. On the other hand, an $N$-separable $N$-partite state is called fully separable and a state which is not 2-separable (or biseparable) is called genuinely multipartite entangled.

This paper utilizes a separability criterion based on purity due to its straightforwardness. Purity is defined as
\begin{equation}
	\label{eq:purity}
	P(\rho) = \text{tr}~\rho^2
\end{equation}
whose value ranges from $1/d$ ($d$ being the dimension of the system) for maximally mixed states and $1$ for pure states. Taking the purity of Eq.~(\ref{eq:nrhoseparable}), we obtain
\begin{equation}
	\label{eq:pursep}
	\text{tr}~\rho^2 = \sum_{i,i'}p_i p_{i'}\text{tr}~(\rho_i^1\cdot\rho_{i'}^1)\cdots\text{tr}~(\rho_i^n\cdot\rho_{i'}^n)
\end{equation}
where orthogonality between states is not assumed.

Now, if we trace out all the other subsystems except one, say the first subsystem $\rho^1$ and obtain the purity, we then have
\begin{equation}
	\label{eq:purone}
	\text{tr}~(\rho^1)^2 = \sum_{i,i'}p_i p_{i'}\text{tr}~(\rho_i^1\cdot\rho_{i'}^1).
\end{equation}
As can be easily seen, the purity in Eq.~(\ref{eq:purone}) is greater than Eq.~(\ref{eq:pursep}) since the factors of the form $\text{tr}~(\rho_i^j\cdot\rho_{i'}^j)$ is bounded by one. In general, we have
\begin{equation}
	\label{eq:ineq}
	\text{tr}~\rho^2 \leq \text{tr}~(\rho^i)^2
\end{equation}
for all subsystems $i$, the equality holding if we assume orthogonality between the states. Since Eq.~(\ref{eq:nrhoseparable}) is the general form of separable states, we then have the following proposition:
\begin{prop}
	If $\rho$ is separable, then $\text{tr}~\rho^2\leq\text{tr}~(\rho^i)^2$ for all the subsystems $i$.
\end{prop}
It is important to note here that the above statement's converse does not hold. In other words, having $\text{tr}~\rho^2\leq\text{tr}~(\rho^i)^2$ true for all the subsystems $i$ does not guarantee that the state $\rho$ is separable. Now, we can take the contrapositive of the above statement as follows:
\begin{prop}
	If there exists a subsystem $i$ such that $\text{tr}~(\rho_i)^2 < \text{tr}~\rho^2$, then $\rho$ is entangled.
\end{prop}

Proposition 2 qualitatively encapsulates the physical implication of what it means to be entangled. To be entangled means that the maximal knowledge of the whole does not necessarily constitute the maximal knowledge of its parts. The ignorance of the subsystem could be qualitatively attributed to the fact that part of it is correlated to the other subsystem, hence looking at the subsystem alone is not enough to gain the entire information about it. The inequality $\text{tr}~(\rho_i)^2 < \text{tr}~\rho^2$ says just that, i.e., if the purity of the whole system (your knowledge of the whole system) is greater than the purity of the subsystem (your knowledge of the subsystem), then this means that the subsystems must be correlated, hence entangled. In other words, you obtain lesser information when looking at the individual subsystems rather than the whole and this missing information is attributed to the correlations of the subsystems with each other. The leap from purity to information (knowledge) is done above since the purity is related to the 2-Renyi entropy given by $S_{\alpha=2}(\rho) = -\frac{1}{2}\log{\text{tr}~(\rho^2)}$ \cite{entropicinequality,information}. Since the 2-Renyi entropy is a monotonic function of the purity, we considered the latter instead of the former for the sake of simplicity. In fact, Proposition 2 is a simpler restatement of the 2-Renyi entropic inequality $S_{\alpha=2}(\rho) \geq S_{\alpha=2}(\rho_i)$.

Also, we use the above criterion as opposed to existing ones since it is very straightforward to apply. Given a state $\rho$, obtain its purity then take the partial traces in order to obtain the purity of the subsystems. If one of the purity of the subsystems turns out to be less than the purity of the whole, then the state $\rho$ is entangled. Note also that the utility of our separability criterion is not limited in any way by the number of systems or the dimensionality of the Hilbert space.

Just like the Renyi entropy, there are certain states where the purity-based separability criterion fails \cite{PPT}. Consider the St\o rmer state \cite{stormer,stormer2} which is a $3\times3$ parameterized by $\alpha$ with values $2\leq\alpha\leq 5$:
\begin{equation}
	\label{eq:stommer}
	\sigma_\alpha = \frac{2}{7}|\psi_+\rangle\langle\psi_+|+\frac{\alpha}{7}\sigma_+ + \frac{5-\alpha}{7}\sigma_-
\end{equation}
where
\begin{equation}
	\label{eq:plus}
	\sigma_+ = \frac{1}{3}(|0\rangle|1\rangle\langle 0|\langle 1| + |1\rangle|2\rangle\langle 1|\langle 2| + |2\rangle|0\rangle\langle 2|\langle 0|)
\end{equation}
\begin{equation}
	\label{eq:minus}
	\sigma_- = \frac{1}{3}(|1\rangle|0\rangle\langle 1|\langle 0| + |2\rangle|1\rangle\langle 2|\langle 1| + |0\rangle|2\rangle\langle 0|\langle 2|)
\end{equation}
\begin{equation}
	\label{eq:psiplus}
	|\psi_+\rangle = \frac{1}{\sqrt{3}}(|00\rangle+|11\rangle+|22\rangle).
\end{equation}
It is known that the St\o rmer state is separable for $2\leq\alpha\leq 3$, bound entangled for $3<\alpha\leq 4$ and free entangled for $4<\alpha\leq 5$. However, simple calculation will show that our criterion detects the entanglement (that is, $P(\sigma_{\alpha,i}) < P(\sigma_\alpha)$) for $\alpha < -1$ and $\alpha > 6$. 

The case considered above is important since it provides a counterexample to the converse of our criterion. In other words if $\rho$ is entangled, then \textit{it is not necessarily true} that there exists a subsystem $i$ such that $\text{tr}~(\rho_i)^2 < \text{tr}~\rho^2$, as is the case when $3<\alpha\leq 5$, whereby $\text{tr}~(\rho_i)^2 \geq \text{tr}~\rho^2$ even though $\sigma_\alpha$ is entangled for these values of $\alpha$. What this entails is that we can actually think of two types of entanglement: (I) entangled states which are detected by Proposition 2 and (II) entangled states which are not detected by Proposition 2. Hence, any entanglement detected using Proposition 2 will fall under type I. These are actually entangled states which are useful in specific nonclassical tasks such as the reduction of communication complexity or quantum cryptography as pointed out in \cite{partial}. In particular, under the context of optimal state merging protocol, the same reference mentioned two regimes which they called the classical (type II in our case where the purity of the whole is less than the purity of its parts) whose partial information $S(A|B)$ is positive and the quantum (type I in our case where the purity of the whole is greater than the purity of its parts) whose partial information $S(A|B)$ is negative. The partial information \cite{horodecki,partial} is defined as
\begin{equation}
S(A|B) = S(\rho_{AB}) - S(\rho_B)
\end{equation}
where $S(\rho)$ is entropy of the state $\rho$.

%Hence, our criterion is again inconclusive for values of $\alpha$ where the St\o rmer state is defined. In other words, even though the St\o rmer state is bound or free entangled for $3<\alpha\leq 5$, our criterion still fails to detect it. This provides a counterexample to the aforementioned conjecture regarding the relationship of our criterion to bound entanglement. That is, our criterion may fail whether the entanglement is bound or free.

Another important thing to note about Proposition 2 is that the entanglement it detects may not be genuine. It is easy to see this by considering the following example. We construct the biseparable 4-partite state $\rho = \rho_{12}\otimes\rho_{34}$ where $\rho_{12} = \rho_{34} = |\psi\rangle\langle\psi|$, and $|\psi\rangle$ could be any of the bipartite entangled Bell states $|\psi_\pm\rangle = \frac{1}{\sqrt{2}}(|00\rangle \pm |11\rangle)$ and $|\phi_\pm\rangle = \frac{1}{\sqrt{2}}(|01\rangle \pm |10\rangle)$. Then it is easy to see that $\text{tr}~(\rho_i)^2 < \text{tr}~\rho^2$ for all $i$. In other words, Proposition 2 detects $\rho$ to be entangled even though this state is obviously biseparable. The reason is that our proposition detects all the entanglements within the system, and having an $i$ such that $\text{tr}~(\rho_i)^2 < \text{tr}~\rho^2$ only means that this subsystem is entangled with some other particle, but not necessarily with all of them.

\section{Entanglement of the mixture of entangled pure states having maximally mixed one-qudit density matrices}
We define the class of entangled pure states, $\mathcal{M}_N(d)$, in the Hilbert space $\mathcal{H}=\otimes_{k=1}^N\mathcal{H}_k=\mathcal{H}_1\otimes\mathcal{H}_2\otimes\dots\otimes\mathcal{H}_N$, where $d=\mbox{dim}\mathcal{H}_k$. Elements of $\mathcal{M}_N(d)$, $\ket{\Phi}$, has the defining property
\begin{equation}
\label{eq:maximal}
\mbox{tr}_{\mathcal{H}\slash\mathcal{H}_k}\denmat{\Phi}= \frac{\mathbb{I}_k}{d}
\end{equation} 
where $\mathbb{I}_k$ is the identity matrix in $\mathcal{H}_k$ and $d$ is the dimension $\mathcal{H}_k$. Note that Eq.~(\ref{eq:maximal}) means that the dimensions of the subsystems must all be equal, that is, $d_1 = \cdots = d_N = d$. The property also implies the normalization condition $\braket{\Phi}{\Phi}=1$. For $N = 2$ and $d = 2$, the states satisfying Eqs.~(\ref{eq:maximal}) are the four Bell states. Some examples of states belonging to $\mathcal{M}_N(d)$ for arbitrary $N$ and $d$ include the $N$-qudit states ($d^{N-1}-1$ of them) given by

\begin{equation}
|\Phi_\text{d}^j\rangle = \frac{1}{\sqrt{d}}\sum_{k=0}^{d-1}|k\rangle^{\otimes j-1}|k+1\rangle|k\rangle^{\otimes N-j}
\end{equation}

\noindent where $j \in \{1,2,3,\dots,d^{N-1}-1\}$ as well as the common $N$-qudit GHZ states given by

\begin{equation}
|\Phi_\text{GHZ}\rangle = \frac{1}{\sqrt{d}}(|0\rangle^{\otimes N} + |1\rangle^{\otimes N} + \cdots + |d\rangle^{\otimes N}).
\end{equation}

We are interested in the mixture of the elements of $\mathcal{M}_N(d)$. We construct the mixed state
\begin{equation}
\label{eq:mixedstate}
\rho(\vec{\lambda}_M,\vec{\Phi})=\sum_{k=1}^M \lambda_k \denmat{\Phi_k}
\end{equation}
where $\vec{\lambda}_M=\{\lambda_1,\dots,\lambda_M\}$ with $\sum_{k=1}^M\lambda_k=1$, $\lambda_k > 0$ and $\vec{\Phi}=\{\ket{\Phi_1},\dots,\ket{\Phi_M}\}$ with $\ket{\Phi_k}\in \mathcal{M}_N(d)$, for some positive integer $M\geq 2$. We refer to the collection of these states as $\widehat{\mathcal{M}}_{M,N}(d)$. When it happens that the vectors in $\vec{\Phi}$ are pairwise orthogonal, that is, $\braket{\Phi_k}{\Phi_l}=\delta_{kl}$, we denote the subset by $\widehat{\mathcal{M}}_{M,N}^{\mbox{ort}}(d)$. The mixed state $\rho$ is a mixture of entangled pure states. However, it is not necessary that the mixture is entangled as can be readily demonstrated. As a trivial example, the concurrence of the mixed state $\rho = \frac{1}{2}|\psi\rangle\langle\psi| + \frac{1}{2}|\phi\rangle\langle\phi|$ is zero, with $|\psi\rangle$ and $|\phi\rangle$ being distinct bipartite entangled Bell states.

We now wish to investigate this class of mixed states by means of the separability criterion given in this paper. We can readily establish the following equalities
\begin{equation}
\mbox{tr}\rho^2 = \sum_{k,l=1}^M \lambda_k \lambda_l \left|\braket{\Phi_k}{\Phi_l}\right|^2
\end{equation}
\begin{equation}
\mbox{tr}_k \rho_k^2 = \frac{1}{d}.
\end{equation}
Then $\rho$ is entangled provided
\begin{equation}
\sum_{k,l=1}^M \lambda_k \lambda_l |\braket{\Phi_k}{\Phi_l}|^2>\frac{1}{d}.
\end{equation}

Also central to the development to follow is the function
\begin{eqnarray}\label{nonort}
P(\vec{\lambda}_M,\vec{\Phi})=\sum_{k,l=1}^M \lambda_k \lambda_l \left|\braket{\Phi_k}{\Phi_l}\right|^2
\end{eqnarray}
which is just the purity of the given mixed state. The ability of the separability criterion to detect entanglement will depend on the minimum of this function that occurs when the involved states are mutually orthogonal and we denote this minimum by
\begin{equation}\label{ort}
P_{\mbox{ort}}(\vec{\lambda}_M,\vec{\Phi})=\sum_{k=1}^M \lambda^2_k = P_{\mbox{ort}}(\vec{\lambda}_M).
\end{equation}

\begin{definition}
	We refer to a subset, $\mathcal{M}$, of $\widehat{\mathcal{M}}_{M,N}(d)$ as entangled if every density matrix $\rho$ of $\mathcal{M}$ is entangled.
\end{definition}

\begin{lemma}
	If $\widehat{\mathcal{M}}_{M,N}^{\mbox{ort}}(d)$ is entangled, then $\widehat{\mathcal{M}}_{M,N}(d)$ is entangled.
\end{lemma}
\begin{proof}
For every $\vec{\lambda}_M$ we have the inequality $P(\vec{\lambda}_M,\vec{\Phi})\geq P_{\mbox{ort}}(\vec{\lambda}_M)$, with equality only when $\vec{\Phi}$ is mutually orthogonal. This follows from the positivity of all the terms in Eq.~(\ref{nonort}) and the fact that $P_{\mbox{ort}}(\vec{\lambda}_M)$ is just the diagonal of $P(\vec{\lambda}_M,\vec{\Phi})$. Since $\widehat{\mathcal{M}}_{M,N}^{\mbox{ort}}(d)$ is, by hypothesis, entangled, we have $P_{\mbox{ort}}(\vec{\lambda}_M)>1/d$ for all $\vec{\lambda}_M$. Then from the inequality $P(\vec{\lambda}_M,\vec{\Phi})\geq P_{\mbox{ort}}(\vec{\lambda}_M)$, we have also $P(\vec{\lambda}_M,\vec{\Phi})>1/d$ for all $\vec{\lambda}_M$ and $\vec{\Phi}$ so that $\widehat{\mathcal{M}}_{M,N}(d)$ is entangled.
\end{proof}

\begin{theorem}
	$\widehat{\mathcal{M}}_{M,N}(d)$ is entangled for all $d>M$.
\end{theorem}
\begin{proof}
	Using the above Lemma it is sufficient to establish that $\widehat{\mathcal{M}}_{M,N}^{\mbox{ort}}(d)$ is entangled when $d>M$. Imposing the constraint $\sum_{k=1}^M\lambda_k=1$, the purity function $P_{\mbox{ort}}(\vec{\lambda}_M)$ assumes the form
	\begin{equation}
	P_{\mbox{ort}}(\vec{\lambda}_M)=\sum_{k=1}^{M-1} \lambda_k^2 + (1-\sum_{k=1}^{M-1}\lambda_k)^2 .
	\end{equation}
	The minimum is obtained by taking its first derivative and equating to zero, $\partial P_{\mbox{ort}}/\partial \lambda_l=0$,
	\begin{equation}
	 \lambda_l +  \sum_{k=1}^{M-1}\lambda_k = 1.
	\end{equation}
	The solution is $\lambda_k=1/M$ for all $k=1,\dots,M$, which we will denote as $\vec{\lambda}_{M,0}$. Then we have the inequality $P_{\mbox{ort}}(\vec{\lambda}_M)\geq 1/M$ for all $\vec{\lambda}_M$. Since $d>M$, then $P_{\mbox{ort}}(\vec{\lambda}_M)> 1/d$. Hence $\widehat{\mathcal{M}}_{M,N}^{\mbox{ort}}(d)$ is entangled and $\widehat{\mathcal{M}}_{M,N}(d)$ is likewise entangled by the above Lemma.
\end{proof}

\begin{theorem}
	$\widehat{\mathcal{M}}_{d,N}(d)/\{\rho(\vec{\lambda}_{M,0},\vec{\Phi}^{\mbox{ort}})\}$ is entangled, where $\vec{\lambda}_{M,0}=\{1/M,\dots,1/M\}$.
\end{theorem}
\begin{proof}
	The purity of the state $\rho(\vec{\lambda}_{M,0},\vec{\Phi}^{\mbox{ort}})$ is $P_{\mbox{ort}}(\vec{\lambda}_{M,0})=1/d$ so that the criteria fails on this state. However, for any $\vec{\lambda}_M\neq \vec{\lambda}_{M,0}$ the purity is necessarily $P_{\mbox{ort}}(\vec{\lambda}_M)>1/d$. This, together with the fact that $P(\vec{\lambda}_{M},\vec{\Phi})>P_{\mbox{ort}}(\vec{\lambda}_{M,0})$ for all $\vec{\Phi}\neq\vec{\Phi}_{\mbox{ort}}$ imply that $\widehat{\mathcal{M}}_{d,N}(d)/\{\rho(\vec{\lambda}_{M,0},\vec{\Phi}^{\mbox{ort}})\}$ is entangled.
\end{proof}

%\subsection{$M=2$}
%First let us consider the case $M=2$ so that the mixed states are given by
%\begin{equation}
%\rho=\lambda \denmat{\Phi_1}+(1-\lambda)\denmat{\Phi_2}
%\end{equation}
%for all $0<\lambda<1$. The purity function is given by
%\begin{equation}
%P=\lambda^2 + (1-\lambda)^2 + 2 \lambda (1-\lambda) |\braket{\Phi_1}{\Phi_2}|^2
%\end{equation}
%The minimum occurs at $\lambda_0=1/2$ and its minimum value is given by
%\begin{equation}
%P_{\mbox{min}}=\frac{1}{2}\left(1+|\braket{\Phi_1}{\Phi_2}|^2\right).
%\end{equation}
%Thus we have the bound $P\geq 1/2$ with equality when $\braket{\Phi_1}{\Phi_2}=0$.
%
%Let consider an $N$-partite qubit system, $\mathcal{H}=\otimes_{k=1}^N\mathbb{C}^2$; for each qubit we have $d=2$. Now when $\braket{\Phi_1}{\Phi_2}=0$ we have $P\geq 1/2$. For $\lambda=1/2$, the purity is $P=1/2$, so that we cannot conclude from the criterion whether the mixed state is entangled or separable. However, for $\lambda\neq 1/2$ we have $P>1/2$ so that the state is entangled. On the other hand, when  $\braket{\Phi_1}{\Phi_2}\neq 0$, we have $P>1/2$ for all $\lambda$, including $\lambda=1/2$. So that the state is entangled.

\section{Genuine entanglement of the mixture of entangled pure states with separable reduced density matrices}
We go back to the caveat at the end of Section 2 that our criterion detects entanglement which may or may not be genuine. In particular, although the class of mixed states that were constructed following the theorems in the previous section are entangled, we could not be sure whether the entanglement is genuine or not. 

Here we define a class of states $\mathcal{N}_N(d)\subset\mathcal{M}_N(d)$. Its elements, which are entangled pure states $|\Phi\rangle$, has the property
\begin{equation}
\label{eq:2separable}
\rho_s = \mbox{tr}_{\mathcal{H}_i}\denmat{\Phi}
\end{equation}
is a separable state. In other words, tracing out a single subsystem leaves a state which is separable. As it happens, this is also one of the criteria for maximal entanglement \cite{gisin,higuchi}, i.e., a measurement on any one of the qudits destroys the entanglement between the remaining ones. This property is meaningful only for states with $N \geq 3$ as should be the case since the issue of genuine entanglement arises only for these multipartite states. 

For $N = 3$ and $d = 2$, these are the states of subtype 2-0 using the classification in \cite{sabin}. For an explicit example, the 4-qubit GHZ state
\begin{equation}
\label{eq:4GHZ}
|GHZ_4^\pm\rangle = \frac{1}{\sqrt{2}}(|0000\rangle \pm |1111\rangle)
\end{equation}
belongs to $\mathcal{N}_N(d)$ while the 4-qubit Dicke state
\begin{align}
\label{eq:42Dicke}
|D_2^4\rangle = \frac{1}{\sqrt{6}}&(|1100\rangle + |1010\rangle + |1001\rangle \\ \nonumber &+ |0110\rangle + |0101\rangle + |0011\rangle)
\end{align}
and the 4-qubit W state
\begin{align}
\label{eq:4W}
|W_4\rangle = \frac{1}{\sqrt{4}}(|0001\rangle + |0010\rangle + |0100\rangle + |1000\rangle)
\end{align}
do not.

As with the previous section, we construct the mixed state
\begin{equation}
\label{eq:mixedstate2}
\rho(\vec{\lambda}_M,\vec{\Phi})=\sum_{k=1}^M \lambda_k \denmat{\Phi_k}
\end{equation}
where $\vec{\lambda}_M=\{\lambda_1,\dots,\lambda_M\}$ with $\sum_{k=1}^M\lambda_k=1$, $\lambda_k > 0$ and $\vec{\Phi}=\{\ket{\Phi_1},\dots,\ket{\Phi_M}\}$ with $\ket{\Phi_k}\in \mathcal{N}_N(d)$, for some positive integer $M\geq 2$. We refer to the collection of these states as $\tilde{\mathcal{N}}_{M,N}(d)$. We then show that the following theorem holds:

\begin{theorem}
Suppose that $\rho \in \tilde{\mathcal{N}}_{M,N}(d)$. If $\text{tr}~(\rho_i)^2 < \text{tr}~\rho^2$ for all $i$, then $\rho$ is genuinely entangled.
\end{theorem}
\begin{proof}
To prove the theorem above, we show that all possible bipartitions of the state $\rho$ remains entangled. In particular, we show that $\rho_j$ is entangled with $\rho_{(j)}$, $\rho_{ij}$ is entangled with $\rho_{(ij)}$, $\rho_{ijk}$ is entangled with $\rho_{(ijk)}$ and so on, where $\rho_{(i)} = \mbox{tr}_{\mathcal{H}_i}\rho$ which means that the $i^\text{th}$ subsystem has been traced out. To show this, note that $\rho_{(j)} = \mbox{tr}_{\mathcal{H}_j}\rho = \sum_k\lambda_k \mbox{tr}_{\mathcal{H}_j}\denmat{\Phi_k}$. However, according to Eq.~(\ref{eq:2separable}), $\mbox{tr}_{\mathcal{H}_j}\denmat{\Phi_k}$ is a separable state. Then $\mbox{tr}~(\rho_{(j)})^2 \leq \mbox{tr}~(\rho_i)^2$ via Proposition 1 which implies that $\mbox{tr}~(\rho_{(j)})^2 < \mbox{tr}~\rho^2$. By Proposition 2, this means that $\rho_{(j)}$ is entangled with $\rho_j$. Similarly, $\rho_{(jk)} = \sum_{k'}\lambda_{k'} \mbox{tr}_{\mathcal{H}/\mathcal{H}_j,\mathcal{H}_k}\denmat{\Phi_{k'}}$. Again, separability of $\mbox{tr}_{\mathcal{H}/\mathcal{H}_j,\mathcal{H}_k}\denmat{\Phi_k}$ implies that $\mbox{tr}~(\rho_{(jk)})^2 \leq \mbox{tr}~(\rho_i)^2 < \mbox{tr}~\rho^2$, which means that the subsystem $\rho_{(jk)}$ is entangled with $\rho_{jk}$. We can do the same procedure to show that all the rest of the possible bipartitions ($3$ and $N-3$, $4$ and $N-4$ and so on back to $N-2$ and $2$) remain entangled. Hence, the entanglement of $\rho$ is genuine.
\end{proof}

Motivated by the above theorem, we define $\widehat{\mathcal{N}}_{M,N}(d)$ as the class of states that belong to $\tilde{\mathcal{N}}_{M,N}(d)$ as well as having the property that $\text{tr}~(\rho_i)^2 < \text{tr}~\rho^2$ for all $i$. Thus, if we want the class of entangled mixed states constructed in Section 2 via Theorems 1 and 2 to have genuine entanglement, then we restrict ourselves to the class of entangled mixed states belonging to $\widehat{\mathcal{N}}_{M,N}(d)$, i.e., the mixed state $\rho(\vec{\lambda}_M,\vec{\Phi})=\sum_{k=1}^M \lambda_k \denmat{\Phi_k}$ such that $\text{tr}~(\rho_i)^2 < \text{tr}~\rho^2$ for all $i$ and $|\Phi_k\rangle\in\mathcal{N}_N(d)$. In other words, $|\Phi_k\rangle$ should satisfy both Eqs.~(\ref{eq:maximal}) and (\ref{eq:2separable}).

As explicit examples, consider the mixture
\begin{align}
\label{eq:rhoA}
\rho^A = \frac{3}{4}|GHZ_4^+\rangle\langle GHZ_4^+| + \frac{1}{4}|GHZ_4^-\rangle\langle GHZ_4^-|.
\end{align}
Here, $\mbox{tr}~(\rho^A)^2 = 0.625$ and $\mbox{tr}~(\rho_i^A)^2 = 0.5$ for all $i$. Hence, $\text{tr}~(\rho_i)^2 < \text{tr}~\rho^2$ for all $i$. Also, $|GHZ_4^\pm\rangle\in\mathcal{N}_N(d)$, hence, $\rho^A$ is entangled according to Theorem 2 and its entanglement is genuine. On the other hand, the mixture
\begin{align}
\label{eq:rhoB}
\rho^B = \frac{3}{4}|GHZ_4^+\rangle\langle GHZ_4^+| + \frac{1}{4}|D_2^4\rangle\langle D_2^4|
\end{align}
also has $\mbox{tr}~(\rho^B)^2 = 0.625$ and $\mbox{tr}~(\rho_i^B)^2 = 0.5$ for all $i$. However, $|D_2^4\rangle\in\mathcal{M}_N(d)$ but $|D_2^4\rangle\notin\mathcal{N}_N(d)$, hence, $\rho^B$ is entangled according to Theorem 2 but its entanglement \textit{may not} be genuine. We emphasize that we are not claiming that the entanglement is not genuine. Rather, we are saying that the criterion is not enough to judge whether the entanglement it detects for states not in $\widehat{\mathcal{N}}_N(d)$ are genuine or not.

In general, Dicke states of the form $|D_{N/2}^N\rangle$, like Eq.~(\ref{eq:42Dicke}) for $N = 4$ belong to $\mathcal{M}_N(d)/\mathcal{N}_N(d)$. To show this, recall that for $N$-qubits with a parameter $m$ between 1 to $n - 1$, a Dicke state is generally written as
\begin{equation}
\label{eq:Dicke}
|D_m^N\rangle = \begin{pmatrix} N\\m\end{pmatrix}^{-\frac{1}{2}}\sum_{|\{\alpha\}|=m}|d_{\{\alpha\}}\rangle
\end{equation}
where the sum runs over all $\{\alpha\}\subset\{1,2,\dots,N\}$ which are sets of $m$ different integers between 1 and $N$ and $|d_{\{\alpha\}}\rangle$ is the product state with $|1\rangle$ in all subsystems and whose numbers are contained in $\{\alpha\}$ and $|0\rangle$ otherwise. Then tracing out all the subsystems except one gives us a state
\begin{equation}
\label{eq:mixedDicke}
\rho_i = \frac{1}{\Omega}\left(\begin{pmatrix} N-1\\m\end{pmatrix}|0\rangle\langle 0| + \begin{pmatrix} N-1\\m-1\end{pmatrix}|1\rangle\langle 1|\right)
\end{equation}
where $\Omega = \begin{pmatrix} N-1\\m\end{pmatrix} + \begin{pmatrix} N-1\\m-1\end{pmatrix}$. For Eq.~(\ref{eq:mixedDicke}) to be maximally mixed, we want $\begin{pmatrix} N-1\\m\end{pmatrix} = \begin{pmatrix} N-1\\m-1\end{pmatrix}$. This implies that $m = \frac{N}{2}$. Now, for Dicke states $\rho = |D_{N/2}^N\rangle\langle D_{N/2}^N|$, tracing out all but two of the subsystems leaves a state $\rho_{ij}$ which has a concurrence of $C(\rho_{ij}) = \frac{1}{N-1}$. Hence, Dicke states satisfy the property in Eq.~(\ref{eq:maximal}) but not Eq.~(\ref{eq:2separable}), in other words, $|D_{N/2}^N\rangle\in\mathcal{M}_N(d)$ but $|D_{N/2}^N\rangle\notin\mathcal{N}_N(d)$.

To summarize this section, if we are only interested in constructing entangled states without regard to whether or not the entanglement is genuine, then we can drop the property in Eq.~(\ref{eq:2separable}), giving us a larger set of pure states that we can mix from. However, if we want to ensure that the constructed mixed states contain genuine entanglement, then we do so at the expense of making the set of pure states we can mix from smaller, and this is done by adding the property in Eq.~(\ref{eq:2separable}).

\section{Consistency with concurrence}

In this section, we will discuss the consistency of our theorems in the previous section with the measure for quantifying bipartite entanglement known as the concurrence. In particular, we will construct states that are entangled according to our theorems and see if it is entangled under concurrence as well. We limit our comparisons to the case of bipartite qubits ($d = 2$ and $N = 2$) for two reasons. First, for the two-qubit case, the concurrence provides a computable formula for the entanglement of formation, a measure based on convex roof. Second, we still don't have a consensus as to what is the correct entanglement measure for states with arbitrary $d$ and $N$. In fact, many multipartite entanglement measure exists that define what it means to be maximally entangled in different ways \cite{vedral,wei, dur, love, gisin, brown, osterloh, gour, helwig}.

Theorem 1 is trivially satisfied by $d = 2$ since it would require us to mix $M = 1$ states out of the Bell states. In general, a mixed state comprised of the Bell states is given by
\begin{align}
\label{eq:mixed2}
\rho~=~&\alpha|\psi_+\rangle\langle\psi_+| + \beta|\psi_-\rangle\langle\psi_-| + \gamma|\phi_+\rangle\langle\phi_+| \nonumber \\ &+(1-\alpha-\beta-\gamma)|\phi_-\rangle\langle\phi_-|
\end{align}
where $\alpha + \beta + \gamma < 1$. Also, without loss of generality, we can assume that $\alpha>\beta>\gamma$. Note that the state has $d < M$ and so it does not fall under the condition required by Theorem 1. Are states not falling under Theorem 1 necessarily separable? We can readily calculate the concurrence of Eq.~(\ref{eq:mixed2}) giving us $C(\rho) = 1$. Thus, we see that there are entangled states which are not part of the class of entangled states given by Theorem 1 (i.e., those that have $d > M$) and these entangled states which are not detected by Theorem 1 automatically belong to that of type II. Note that if we let $\alpha = \beta = \gamma = \frac{1}{4}$, then we find that $C(\rho) = 0$.

Now, let us investigate Theorem 2 for bipartite qubits. Here, Theorem 2 requires that $d = M = 2$ and so in general, the mixed state will be given by
\begin{equation}
\label{eq:mixed=M}
\rho = \lambda|\phi_1\rangle\langle\phi_1| + (1-\lambda)|\phi_2\rangle\langle\phi_2|
\end{equation}
where $|\phi_1\rangle$ and $|\phi_2\rangle$ could be any of the distinct Bell states. Calculation of the concurrence of Eq.~(\ref{eq:mixed=M}) yields $C(\rho) = |2\lambda - 1|$ which is equal to zero only for $\lambda = \frac{1}{2}$. In other words, mixing $d$ number of entangled $d$-dimensional pure states that are not equally weighted automatically yields an entangled mixed state and this state belongs to that of type I. We conjecture that this holds for a general $N$-partite qudit systems although we can only check it for the bipartite qubit case since a generalized measure for an $N$-partite qudit mixed state (one that does not involve optimization) is not yet available.

\section{Conclusion}

We have considered mixtures of entangled pure states with maximally mixed one-qudit reduced density matrices and studied the conditions were they remain entangled using purity as our separability criterion. We have found that in order for the resulting mixed state to remain entangled, then the number of pure states with maximally mixed one-qudit reduced density matrices to be mixed must be less than or equal to its dimension. For the $d < M$ case, we found that there are entangled states which are not detected although these entangled states belong to type II which are ``undesirable'' in terms of their utility for the reduction of communication complexity or quantum cryptography. For the $d = M$ case, we found that the resulting mixed states are entangled provided that the entangled pure states with maximally mixed one-qudit reduced density matrices to be mixed are not equally weighted. We have shown that it is consistent with what is predicted by the concurrence for the case of bipartite qubits. However, such comparison can't be made for the general case of mixed $N$-partite qudits due to a lack of computable measure similar to concurrence. Also, we've shown that we can obtain genuinely entangled mixed states at the expense of restricting the set of entangled pure states that we can mix from.

However, there still exist open problems for future research that will be considered elsewhere like the physical significance of the relationship between $d$ and $M$ as well as an understanding of the $d > M$ case where the criterion fails. Also, it will be fruitful to look at much ``stronger'' separability criteria where the entropic-based ones fail. Another interesting case to consider would be the mixing of pure states of different dimensions and its effect on the entanglement of the resulting mixed state.


\begin{thebibliography}{00}

\bibitem{detection}
O. Guhne, G. Toth, Phys. Rep. \textbf{474}, 1 (2009).

\bibitem{horodecki}
R. Horodecki, P. Horodecki, M. Horodecki, and K. Horodecki, Rev. Mod. Phys. \textbf{81}, 865 (2007).

\bibitem{quantifying}
C. Eltschka, J. Siewert, J. Phys. A: Math. Theor. \textbf{47}, 424005 (2014).

\bibitem{entropy}
A. Uhlmann, Entropy \textbf{12}, 1799 (2010).

\bibitem{popescu}
S. Popescu, Phys. Rev. Lett. \textbf{72}, 6 (1994).

\bibitem{partial}
M. Horodecki, J. Oppenheim, and A. Winter, Nature, \textbf{436}, 673 (2005).

\bibitem{higuchi}
A. Higuchi, A. Sudbery, Phys. Lett. A \textbf{273}, 213 (2000).

\bibitem{sabin}
C. Sabin, G. Garcia-Alcaine, Eur. Phys. J. D, \textbf{48}, 435 (2008).

\bibitem{PPT}
A. Peres, Phys. Rev. Lett. \textbf{77}, 8 (1996).

\bibitem{stormer}
P. Horodecki, R. Horodecki, Quant. Inf. and Comp. \textbf{1}, 1 (2001).

\bibitem{stormer2}
P. Horodecki, M. Horodecki, and R. Horodecki, Phys. Rev. Lett. \textbf{82}, 5 (1999).

\bibitem{entropicinequality}
R. Horododecki, P. Horodecki, and M. Horodecki, Phys. Lett. A, \textbf{210}, 377 (1996).

\bibitem{information}
R. Horodecki, M. Horodecki, Phys. Rev. A, \textbf{54}, 3 (1996).

\bibitem{wootters}
W. Wootters, Phys. Rev. Lett. \textbf{80}, 2245, eprint quant-ph/97099029, (1998).

\bibitem{vedral}
V. Vedral, Nature \textbf{453}, 07124 (2008).

\bibitem{wei}
T. Wei, K. Nemoto, P. M. Goldbart, P. G. Kwiat, W. J. Munro and F. Verstraete, Phys. Rev. A, \textbf{67}, 022110 (2003).

\bibitem{dur}
W. Dür, G. Vidal and J. I. Cirac, Phys. Rev. A \textbf{62}, 062314 (2000).

\bibitem{love}
P. Love, A. Brink, A. Smirnov, M. Amin, M. Grajcar, E. Ilichev, A. Izmalkov and A. Zagoskin, Quant. Inf. Process., \textbf{6}, 187 (2007).

\bibitem{gisin}
N. Gisin and H. Bechmann-Pasquinucci, Phys. Lett. A, \textbf{246}, 1 (1998).

\bibitem{brown}
I. D. K. Brown, S. Stepney, A. Sudbery and S. L. Braunstein, J. Phys. A: Math. Gen. \textbf{38}, 1119 (2005).

\bibitem{osterloh}
A. Osterloh and J. Siewert, New J. Phys. \textbf{12}, 075025 (2010).

\bibitem{gour}
G. Gour and N. R. Wallach, J. Math. Phys. \textbf{51}, 112201 (2010).

\bibitem{helwig}
W. Helwig, W. Cui, J. I. Latorre, A. Riera and H. Lo, Phys. Rev. A \textbf{86}, 052335 (2012).

\end{thebibliography}
\end{document}